\documentclass[12pt]{article}

\usepackage{fullpage}
\usepackage{amsthm,amssymb,amsmath}  
\usepackage{xspace,enumerate}
\usepackage[utf8]{inputenc}
\usepackage{thmtools}
\usepackage{thm-restate}
\usepackage{authblk}

  \theoremstyle{plain}
  \newtheorem{theorem}{Theorem}
  \newtheorem{lemma}[theorem]{Lemma}  
  \newtheorem{corollary}[theorem]{Corollary}  
  \newtheorem{fact}[theorem]{Fact}
  \newtheorem{observation}[theorem]{Observation}
  \theoremstyle{definition}
  
  \newtheorem{problem}[theorem]{Problem}

\title{Approximating longest common substring with $k$~mismatches: Theory and practice\thanks{This is a full version of~\cite{LCSk}.}}
\author[1]{Garance Gourdel}
\author[2]{Tomasz Kociumaka}
\author[3]{Jakub Radoszewski}
\author[4]{Tatiana Starikovskaya}

\affil[1]{ENS Saclay\\
    \texttt{garance.gourdel@ens-paris-saclay.fr}
}
\affil[2]{Bar-Ilan University\\
    \texttt{kociumaka@mimuw.edu.pl}
}
\affil[3]{Institute of Informatics, University of Warsaw \& Samsung R\&D Institute\\
    \texttt{jrad@mimuw.edu.pl}
}
\affil[4]{DI/ENS, PSL Research University\\
    \texttt{tat.starikovskaya@gmail.com}
}

\date{\vspace{-5ex}}

\usepackage{algorithm,algorithmicx,algpseudocode}
\usepackage{microtype}
\usepackage{xspace}
\usepackage{hyperref}
\usepackage{graphicx}
\usepackage{multirow}
\usepackage{diagbox}
\usepackage{array}
\usepackage{subcaption}
\usepackage{tikz}
\usetikzlibrary{decorations.pathreplacing}
\usepackage{xcolor}

\newtheorem*{hypothesis}{\bf Hypothesis}

\newcommand{\Oh}{\mathcal{O}}
\newcommand{\eps}{\varepsilon}

\newcommand{\lcpk}{\mathrm{LCP}_{k}}
\newcommand{\lcsk}{\mathrm{LCS}_{k}}
\newcommand{\lcske}{\mathrm{LCS}_{(1+\eps)k}}
\newcommand{\lcsak}{\mathrm{LCS}_{\tilde{k}}}
\newcommand{\sk}{\mathrm{sk}}
\newcommand{\Prob}{\mathrm{Pr}}

\newcommand{\kLCS}{\textsf{LCS with $k$ Mismatches}\xspace}
\newcommand{\kApproxLCS}{\textsf{LCS with Approximately $k$ Mismatches}\xspace}
\newcommand{\Bichromatic}{\textsf{$(1+\gamma)$-approximate Bichromatic Closest Pair}\xspace}
\newcommand{\NN}{\textsf{Approximate Near Neighbour}\xspace}
\newcommand{\twentyquestions}{\textsf{Twenty Questions}\xspace}
\newcommand{\Carole}{\mathit{Carole}}
\newcommand{\pop}{\mathit{pop}}
\newcommand{\push}{\mathit{push}}
\newcommand{\ttop}{\mathit{top}}
\newcommand{\mmid}{\mathit{mid}}

\newcommand{\Hashes}{\mathcal{H}}
\newcommand{\Collisions}{C}

\newcommand{\Projections}{\Pi}

\newcommand{\HD}{d_H}

\newcommand\restr[2]{{
  \left.\kern-\nulldelimiterspace
  #1 
  \vphantom{\big|} 
  \right|_{#2} 
}}

\newcommand{\norm}[1]{\ensuremath{\lVert#1\rVert}}

\newcommand{\ceil}[1]{\ensuremath{\lceil#1\rceil}}

\begin{document}
  \maketitle

\begin{abstract}
In the problem of the longest common substring with $k$ mismatches we are given two strings $X, Y$ and must find the maximal length $\ell$ such that there is a length-$\ell$ substring of $X$ and a length-$\ell$ substring of $Y$ that differ in at most $k$ positions. The length $\ell$ can be used as a robust measure of similarity between $X, Y$. In this work, we develop new approximation algorithms for computing $\ell$ that are significantly more efficient that previously known solutions from the theoretical point of view. Our approach is simple and practical, which we confirm via an experimental evaluation, and is probably close to optimal as we demonstrate via a conditional lower bound.
\end{abstract}

\section{Introduction}\label{sec:intro}
For decades, the edit distance and its variants remained the most relevant measure of similarity between biological sequences. However, there is strong evidence that the edit distance cannot be computed in strongly subquadratic time~\cite{DBLP:conf/stoc/BackursI15}. 
One possible approach to overcoming the quadratic time barrier is computing the edit distance approximately, and last year in the breakthrough paper Chakraborty et al.~\cite{DBLP:conf/focs/ChakrabortyDGKS18} showed a constant-factor approximation algorithm that computes the edit distance between two strings of length $n$ in time $\tilde{\Oh}(n^{2-2/7})$. Nevertheless, the algorithm is highly non-trivial and because of that is likely to be impractical. 

A different approach is to consider alignment-free measures of similarities. Ideally, we want the measure to be robust and simple enough so that we could compute it efficiently. One candidate for such a measure is the length of the longest common substring with $k$ mismatches. Formally, given two strings $X,  Y$ of lengths at most $n$ and an integer $k$, we want to find the maximal length $\lcsk(X,Y)$ of a substring of $X$ that occurs in $Y$ with at most $k$ mismatches. Computing this value constitutes the \kLCS problem.

The \kLCS problem was first considered for $k = 1$~\cite{DBLP:journals/poit/BabenkoS11,DBLP:journals/ipl/FlouriGKU15}, with current best algorithm taking $\Oh(n \log n)$ time and $\Oh(n)$ space. The first algorithm for the general value of $k$ was shown by Flouri et al.~\cite{DBLP:journals/ipl/FlouriGKU15}. Their simple approach used quadratic time and linear space. Grabowski~\cite{DBLP:journals/ipl/Grabowski15} focused on a data-dependent approach, namely, he showed two linear-space algorithms with running times $\Oh (n ((k+1) (\mathrm{LCS}+1))^k)$ and $\Oh (n^2 k/\lcsk)$, where $\mathrm{LCS}$ is the length of the longest common substring of $X$ and~$Y$ and $\lcsk$, similarly to above, is the length of the longest common substring with $k$ mismatches of $X$ and $Y$. Abboud et al.~\cite{DBLP:conf/soda/AbboudWY15} showed a $k^{1.5} n^2 / 2^{\Omega(\sqrt{(\log n)/k})}$-time randomised solution to the problem via the polynomial method. Thankachan et al.~\cite{DBLP:journals/jcb/ThankachanAA16} presented an $\Oh(n \log^k n)$-time, $\Oh(n)$-space solution for constant $k$. This approach was recently extended by Charalampopoulos et al.~\cite{DBLP:conf/cpm/Charalampopoulos18} to develop an $\Oh(n)$-time  and $\Oh(n)$-space algorithm for the case of $\lcsk = \Omega(\log^{2k+2} n)$.

On the other hand, Kociumaka, Radoszewski, and Starikovskaya~\cite{DBLP:journals/algorithmica/KociumakaRS19} showed that there is $k = \Theta(\log n)$ such that the \kLCS problem cannot be solved in strongly subquadratic time, even for the binary alphabet, unless the Strong Exponential Time Hypothesis (SETH) of Impagliazzo, Paturi, and Zane~\cite{DBLP:journals/jcss/ImpagliazzoPZ01} is false. This conditional lower bound implies that there is little hope to improve existing solutions to \kLCS. To overcome this barrier, they introduced an approximation approach to \kLCS, inspired by the work of Andoni and Indyk~\cite{substringNN}. 

\begin{problem}[\kApproxLCS]\label{pr:LCS'k}
Two strings $X, Y$ of length at most~$n$, an integer $k$, and a constant $\eps > 0$ are given. Return a substring of $X$ of length at least $\lcsk(X,Y)$ that occurs in $Y$ with at most $(1+\eps) \cdot k$ mismatches.
\end{problem}

Kociumaka, Radoszewski, and Starikovskaya~\cite{DBLP:journals/algorithmica/KociumakaRS19} also showed that for any $\eps \in (0,2)$ the \kApproxLCS problem can be solved in $\Oh (n^{1+1/(1+\eps)} \log^2 n)$ time and $\Oh (n^{1+1/(1+\eps)})$ space. Besides for superlinear space, their solution uses a very complex class of hash functions which requires $n^{4/3+o(1)}$-time preprocessing, and that is the underlying reason for the bounds on~$\eps$. In this work, we significantly improve the complexity of the \kApproxLCS problem and show the following results.

\begin{theorem}\label{th:klcs_upper}
Let $\eps > 0$ be an arbitrary constant. The \kApproxLCS problem can be solved correctly with high probability:
\begin{enumerate}[1)]
\item In $\Oh(n^{1+ 1/(1+2\eps) + o(1)})$ time and $\Oh(n^{1+ 1/(1+2\eps) + o(1)})$ space assuming a constant-size alphabet;
\item In $\Oh(n^{1+1/(1+\eps)} \log^3 n)$ time and $\Oh(n)$ space for alphabets of arbitrary size. 
\end{enumerate}
\end{theorem}

Our first solution uses the Approximate Nearest Neighbour data structure~\cite{DBLP:conf/stoc/AndoniR15} as a black box. The definition of this data structure is extremely involved, and we view this result as more of a theoretical interest.
On the other hand, our second solution is simple and practical, which we confirm by experimental evaluation (see Section~\ref{sec:implem} for details). 

As a final remark, we note that a construction similar to the one used to show a lower bound for the \kLCS problem~\cite{DBLP:journals/algorithmica/KociumakaRS19} gives a lower bound for \kApproxLCS. 

\begin{fact}\label{lm:klcs_lower}
Assuming SETH, for every constant $\delta > 0$, there exists a constant $\eps = \eps(\delta)$\footnote{Here $\delta$ is a function of $\eps$ for which the explicit form is not known (a condition inherited from~\cite{DBLP:journals/corr/abs-1803-00904}).} such that any randomised algorithm that solves the \kApproxLCS problem for given $X$ and $Y$ of length at most $n$ correctly with constant probability uses $\Omega(n^{2-\delta})$ time. 
\end{fact}

\subparagraph*{Related work.}
In 2014, Leimester and Morgenstern~\cite{kmacs} introduced a related similarity measure, \emph{the $k$-macs distance}. Let $\lcpk (X_i, Y_j) = \max\{\ell: \HD(X[i,i+\ell-1], Y[j, j+\ell-1]) \le k\}$, where $\HD$ stands for Hamming distance, i.e.\ the number of mismatches between two strings. We have $\lcsk = \max_{i,j} \lcpk(X_i, Y_j)$. The $k$-macs distance, on the other hand, is defined as a normalised average of these values.  Leimeister and Morgenstern~\cite{kmacs} showed a heuristic algorithm for computing the $k$-macs distance, with no theoretical guarantees for the precision of the approximation; other heuristic approaches for computing the $k$-macs distance include~\cite{DBLP:journals/jcb/ThankachanCLAA16,Thankachan2017}. The only algorithm with provable theoretical guarantees is~\cite{DBLP:journals/jcb/ThankachanAA16} and it computes the $k$-macs distance in $\Oh(n \log^k n)$ time and $\Oh(n)$ space.

\section{Preliminaries}\label{sec:prelim}
We assume that the alphabet of the strings $X, Y$ is $\Sigma = \{1,\ldots,\sigma\}$, where $\sigma = n^{\Oh(1)}$.

\subparagraph*{Karp--Rabin fingerprints. } The Karp--Rabin fingerprint~\cite{DBLP:journals/ibmrd/KarpR87} of a string $S = s_1 s_2 \dots s_\ell$ is defined as

$$\varphi(S) = \left(\sum_{i = 1}^{\ell} r^{i-1} s_i\right) \bmod q,$$ 
where $q = \Omega(\max\{n^5, \sigma\})$ is a prime number, and $r \in \mathbb{F}_q$ is chosen uniformly at random.  Obviously, if $S_1 = S_2$, then $\varphi(S_1) = \varphi(S_2)$. Furthermore, for any $\ell \le n$, if the fingerprints of two $\ell$-length strings $S_1, S_2$ are equal, then $S_1, S_2$ are equal with probability at least $1-1/n^4$ (for a proof, see e.g.~\cite{Porat:09}).

\subparagraph*{Dimension reduction. } We will exploit a computationally efficient variant of the Johnson--Lindenstrauss lemma~\cite{MR737400} which describes a low-distortion embedding from a high-dimensional Euclidean space into a low-dimensional one. Let $\norm{\cdot}$ be the Euclidean ($L_2$) norm of a vector. We will exploit the following claim which follows immediately from~\cite[Theorem 1.1]{ACHLIOPTAS2003671}:
 
\begin{lemma}\label{lm:dim_reduction}
Let $P$ be a set of $n$ vectors in $\mathbb{R}^{\ell}$, where $\ell \le n$. Given $\alpha = \alpha(n) > 0$ and a constant $\beta > 0$, there is $d = \Theta(\alpha^{-2}\log n)$ and a scalar $c > 0$ such that the following holds. Let $M$ be a $d \times \ell$ matrix filled with i.u.d.\ $\pm1$ random variables. For all $U \in P$, define $\sk_\alpha (U) = c \cdot M U$. Then for all $U,V \in P$ there is $\norm{U-V}^2 \leq \norm{\sk_\alpha(U)-\sk_\alpha(V)}^2 \leq (1+\alpha) \norm{U-V}^2$ with probability at least $1- n^{-\beta}$. 
\end{lemma}

Since the Hamming distance between binary strings $U, V$ is equal to $\norm{U-V}^2$, the matrix~$M$ defines a low-distortion embedding from an $\ell$-dimensional into a $d$-dimensional Hamming space as well. For non-binary strings, an extra step is required. Let the alphabet be $\Sigma = \{1, 2, \ldots, \sigma\}$ and consider a morphism $\mu : \Sigma \rightarrow \{0,1\}^\sigma$, where $\mu(a) = 0^{a-1} 1 0^{\sigma-a}$ for all $a \in \Sigma$. We extend $\mu$ to strings in a natural way. Note that for two strings $U, V$ over the alphabet $\Sigma$ the Hamming distance between $\mu(U), \mu(V)$ is exactly twice the Hamming distance between $U, V$. We therefore obtain:

\begin{corollary}\label{cor:dim_reduction}
Let $P$ be a set of $n$ strings in $\Sigma^{\ell}$, where $\ell \le n$. Given $\alpha = \alpha(n) > 0$ and a constant $\beta > 0$, there is $d = \Theta(\alpha^{-2}\log n)$ and a scalar $c > 0$ such that the following holds. Let $M$ be a $d \times (\sigma \cdot \ell)$ matrix filled with i.u.d.\ $\pm1$ random variables. For all $U \in P$, define $\sk_\alpha(U) = c\cdot M \mu(U)$. Then for all $U, V \in P$ there is $\HD(U,V) \leq \norm{\sk_\alpha(U) - \sk_\alpha(V)}^2 \leq (1+\alpha) \HD(U, V)$ with probability at least $1- n^{-\beta}$.
\end{corollary}

We will use the corollary for dimension reduction, and also to design a simple test that checks whether the Hamming distance between two strings is at most $k$. 

\begin{corollary}
Let $P$ be a set of $n$ strings in $\Sigma^{\ell}$, where $\ell \le n$. With probability at least $1- n^{-\beta}$, for all $U, V \in P$:
\begin{enumerate}[1)]
\item if $\norm{\sk_\alpha(U)-\sk_\alpha(V)}^2 \le (1+\alpha) k$, then $\HD(U, V) \le (1+\alpha) \cdot k$;
\item if $\norm{\sk_\alpha(U)-\sk_\alpha(V)}^2 > (1+\alpha) k$, then $\HD(U, V) \ge k$.
\end{enumerate}
\end{corollary}

\subsection{The Twenty Questions game}
\label{sec:20questions}
Consider the following version of the classic game ``Twenty Questions''. There are two players: Paul and Carole; Carole thinks of two numbers $A, B$ between $0$ and $N$, and Paul must return some number in $[A,B]$. He is allowed to ask questions of form ``Is $x \le A$?'', for any $x \in [0,N]$. If $x \le A$, Carole must return YES; If $A < x \le B$, she can return anything; and if $B < x$, she must return NO. Paul must return the answer after having asked at most $Q$ questions where Carole can tell at most $\ceil{ \rho Q }$ lies, and only in the case when $x \le A$.  

We show that Paul has a winning strategy for $Q = \Theta (\log n)$ and any $\rho < 1/3$ by a black-box reduction to the result of Dhagat, G{\'a}cs, and Winkler~\cite{Dhagat:1992:PLQ:139404.139409} who showed a winning strategy for $A = B$.

\begin{theorem}[\cite{Dhagat:1992:PLQ:139404.139409}]
For $A = B$, Paul has a winning strategy for all $\rho < \frac{1}{3}$ asking $Q = \ceil{ \frac{8 \log N}{(1-3\rho)^2} }$ questions.
\end{theorem}

This result is obtained by maintaining a stack of trusted intervals. Once Paul knows that $A$ is between $\ell$ and $r$, where $\ell \leq r$, he checks whether $A$ is in the left or the right half of the interval $[\ell,r]$. If no inconsistencies appear (like $A < \ell$ or $r < A$), he pushes the new interval to the stack, else he removes the interval $[\ell,r]$ from the stack of trusted intervals. After $Q$ rounds, Paul returns the only number in the top interval in the stack, which is guaranteed to have length $1$ and to contain $A$. We give the pseudocode of Paul's strategy in Algorithm~\ref{alg:20Q}. By $\Carole(x)$, we denote the answer of Carole for a question ``Is $x \le A$?''.

\begin{algorithm}
\caption{The \twentyquestions game}
\begin{algorithmic}[1]
\State $Q \gets  \ceil{ \frac{8 \log N}{(1-3\rho)^2} }$
\State $S \gets \{[0,N]\}$
\For {$i = 1, 2, \ldots , Q/2$ }
	\State $I = [\ell,r] \gets S.\ttop()$
	\State $\mmid \gets \ceil{ \frac{\ell + r}{2} }$
	\If {$\Carole(mid)$}
	    \If {$\Carole(r)$}
	         $S.\pop()$ \Comment{The answer is inconsistent with $I$; remove $I$.}
	    \Else~$S.\push([\mmid, r])$
        \EndIf
    \Else
        \If {$\Carole(\ell)$}
	         $S.\push([\ell,\mmid-1])$
	    \Else~$S.\pop()$ \Comment{The answer is inconsistent with $I$; remove $I$.}
        \EndIf
	\EndIf
\EndFor
\end{algorithmic}
\label{alg:20Q}
\end{algorithm}

We now a show a winning strategy for our variant of the game. 

\begin{corollary}
\label{cor:twentyquestions}
For $A \le B$, Paul has a winning strategy for all $\rho < \frac{1}{3}$ asking $Q = { \frac{8 \log N}{(1-3\rho)^2} }$ questions.
\end{corollary}
\begin{proof}
We introduce just one change to Algorithm~\ref{alg:20Q}, namely, we return the argument of the largest YES obtained in the course of the algorithm. From the problem statement it follows that the answer is at most $B$. We shall now prove that the answer is at least $A$. If Carole ever returned YES for $A < x \le B$, then it is obviously the case. Otherwise, Carole actually behaved as if she had $A=B$ in mind: apart from the small fraction of erroneous answers, she returned YES for $x \le A$, and NO for $x > A$. Thus, the strategy of Dhagat, G{\'a}cs, and Winkler ends up with $A$ as the answer (and this must be due to a YES for $x = A$).
\end{proof}

\section{\texorpdfstring{\kApproxLCS}{LCS with approximately k mismatches}}\label{sec:klcs}
In this section, we prove Theorem~\ref{th:klcs_upper}. Let us first introduce a decision variant of the \kApproxLCS problem. 

\begin{problem}\label{pr:LCS'k-decision}
Two strings $X, Y$ of length at most $n$, integers $k, \ell$, and a constant $\eps > 0$ are given. We must return:
\begin{enumerate}
\item YES if $\ell \le \lcsk(X,Y)$;
\item Anything if $\lcsk(X,Y) < \ell \le \mathrm{LCS}_{(1+\eps)k}(X,Y)$; 
\item NO if $\mathrm{LCS}_{(1+\eps)k}(X,Y) < \ell$.
\end{enumerate}
If we return YES, we must also give a \emph{witness pair} of length-$\ell$ substrings $S_1$ and $S_2$ of $X$ and $Y$, respectively, such that $\HD(S_1,S_2)\le (1+\eps)k$. 
\end{problem}

The decision variant of the \kApproxLCS problem can be reduced to the following $(c,r)$-\NN problem. 

\begin{problem}\label{pr:NN}
In the $(c,r)$-\NN problem with failure probability~$f$, the aim is, given a set $P$ of $n$ points in $\mathbb{R}^d$, to construct a data structure supporting the following queries: given any point $q\in \mathbb{R}^d$, if there exists $p \in P$ such that $\norm{p-q} \leq r$, then return some point $p' \in P$ such that $\norm{p'-q} \leq cr$ with probability at least $1-f$.
\end{problem}

Using the reduction, we will show our first solution to the \kApproxLCS decision problem based on the result of Andoni and Razenshteyn~\cite{DBLP:conf/stoc/AndoniR15}, who showed that for any constant $f$, there is a data structure for the $(c,r)$-\NN problem that has $\Oh(n^{1+\rho+o(1)}+d \cdot n)$ size, $\Oh(d\cdot n^{\rho+o(1)})$ query time, and $\Oh(d \cdot n^{1+\rho+o(1)})$ preprocessing time, where $\rho = 1/(2c^2-1)$.

\begin{lemma}\label{lm:opt_NN}
Assume an alphabet of constant size $\sigma$. The decision variant of the \kApproxLCS problem can be solved in space $\Oh(n^{1+ 1/(1+2\eps) + o(1)})$ and time $\Oh(n^{1+ 1/(1+2\eps) + o(1)})$. The answer is correct with constant probability. 
\end{lemma}
\begin{proof}
Let $P$ be the set of all length-$\ell$ substrings of $X$ and $Q$ be the set of all length-$\ell$ substrings of $Y$,
all encoded in binary using the morphism $\mu$ (see Section~\ref{sec:prelim}). We start by applying the dimension reduction procedure of Corollary~\ref{cor:dim_reduction} to $P$ and $Q$ with $\alpha = 1/(\log\log n)^{\Theta(1)}$ and $\beta = 2$ to obtain sets $P'$ and $Q'$. We can implement the procedure in $\Oh(\sigma n \log^2 n (\log\log n)^{\Theta(1)}) = \Oh(n \log^{2+o(1)} n)$ time by encoding $X, Y$ using $\mu$ and running the FFT algorithm~\cite{FischerPaterson} for each of the $\Oh(\log^{1+o(1)} n)$ rows of the matrix and $\mu(X), \mu(Y)$. 

To solve the decision variant of \kApproxLCS, we build the data structure of Andoni and Razenshteyn~\cite{DBLP:conf/stoc/AndoniR15} for $(\sqrt{(1+\eps)(1-\alpha)}, \sqrt{(1+\alpha)k})$-\NN over~$Q'$. We make a query for each string in $P'$. If, queried for $\sk_{\alpha}(S_1)\in P'$, where $S_1$ is a length-$\ell$ substring of $X$, the data structure outputs $\sk_{\alpha}(S_2)\in Q'$, where $S_2$ is a length-$\ell$ substring of $Y$, then we compute $\norm{\sk_{\alpha}(S_1)- \sk_{\alpha}(S_2)}^2$. If it is at most $(1+\eps)k$, we output YES and the witness pair $(S_1,S_2)$ of substrings. As the length of vectors in $P'$, $Q'$ is $d = \Oh(\log^{1+o(1)} n)$, we obtain the desired complexity. 

To show that the algorithm is correct, suppose that there are length-$\ell$ substrings $S_1$ and $S_2$ of $X$ and $Y$, respectively, with $\HD(S_1, S_2) \le k$. By Corollary~\ref{cor:dim_reduction}, $\norm{\sk_\alpha(S_1),\sk_\alpha(S_2)}\le \sqrt{(1+\alpha)k}$ holds with probability at least $1-1/n$. Then, when querying for $\sk_\alpha(S_1)$, with constant probability the data structure will output a string $\sk_\alpha(S'_2)$ such that $\norm{\sk_\alpha(S_1) - \sk_\alpha(S'_2)}^2 \le (1+\eps)(1-\alpha^2) k \le (1+\eps) k$. Then, our algorithm will return YES. 

On the other hand, if we output YES with a witness pair $(S_1,S_2)$, then $\norm{\sk_\alpha(S_1) - \sk_\alpha(S_2)}^2\le (1+\eps)k$ implies $\HD(S_1,S_2)\le (1+\eps)k$ with high probability by Corollary~\ref{cor:dim_reduction}.
\end{proof}

While this solution is very fast, it uses quite a lot of space. Furthermore, the data structure of~\cite{DBLP:conf/stoc/AndoniR15} that we use as a black box applies highly non-trivial techniques. To overcome these two disadvantages, we will show a different solution based on a careful implementation of ideas first introduced in~\cite{substringNN} that showed a data structure for approximate text indexing with mismatches. In~\cite{DBLP:journals/algorithmica/KociumakaRS19}, the authors developed these ideas further to show an algorithm that solves the \kApproxLCS problem in $\Oh(n^{1+1/(1+\eps)} )$ space and $\Oh (n^{1+1/(1+\eps)} \log^2 n)$ time for $\eps\in(0,2)$ with constant error probability. In this work, we significantly improve and simplify the approach to show the following result:  

\begin{theorem}\label{th:LCS'k-decision}
Assume an alphabet of arbitrary size $\sigma = n^{\Oh(1)}$. The decision variant of \kApproxLCS can be solved in $\Oh(n^{1+1/(1+\eps)} \log^2 n)$ time and $\Oh(n)$ space. The answer is correct with constant probability. 
\end{theorem}

Let us defer the proof of the theorem until Section~\ref{sec:decision} and start by explaining how we use Lemma~\ref{lm:opt_NN} and Theorem~\ref{th:LCS'k-decision} and the \twentyquestions game to show Theorem~\ref{th:klcs_upper}.

\begin{proof}[Proof of Theorem~\ref{th:klcs_upper}]
We will rely on the modified version of the \twentyquestions game that we
  described in Section~\ref{sec:20questions}. In our case, $A = \lcsk(X,Y)$ and
  $B = \lcske(X,Y)$. For Carole, we use either the algorithm of
  Lemma~\ref{lm:opt_NN}, or the algorithm of Theorem~\ref{th:LCS'k-decision}, 
  with an additional procedure verifying the witness pair $(S_1,S_2)$ character by character to check that it indeed satisfies $\HD(S_1,S_2)\le (1+\eps)k$.
  We output the longest pair of (honest) witness
  substrings found across all iterations. We will return a correct answer
  assuming that the fraction of errors is $\rho <\frac13$. Recall that the algorithm solves the decision variant of the \kApproxLCS problem incorrectly with probability not exceeding a constant $\delta$, and we can ensure $\delta < \frac13$ by repeating it a constant number of times. It means that Carole can answer an individual question erroneously with probability less than~$\frac13$. Therefore, for a sufficiently large constant in the number of queries $Q = \Theta(\log n)$, the fraction of erroneous answers is $\rho < \frac13$ with high probability by Chernoff--Hoeffding bounds. The claim of the theorem follows immediately from Lemma~\ref{lm:opt_NN} and Theorem~\ref{th:LCS'k-decision}.
\end{proof}

\subsection{Proof of Theorem~\ref{th:LCS'k-decision}}\label{sec:decision}
We first give an algorithm for the decision version of the \kApproxLCS problem that uses $\Oh(n \log n)$ space and $\Oh(n^{1+1/(1+\eps)} \log n + \sigma n \log^2 n)$ time, and then we improve the space and time complexity. 

We assume to have fixed a Karp--Rabin fingerprinting function $\varphi$ for a prime $q = \Omega(\max\{n^5, \sigma\})$ and an integer $r \in \mathbb{Z}_q$. With error probability inverse polynomial in $n$, we can find such $q$ in $\Oh(\log^{\Oh(1)}n)$ time;
see~\cite{DBLP:journals/moc/TaoCH12,agrawal2004primes}. 

Let~$\Projections$ be the set of all projections of strings of length $\ell$ onto a single position, i.e., the value $\pi_i(S)$ of the $i$-th projection on a string $S$ of length $\ell$ is simply its $i$-th character $S[i]$. More generally, for a length-$\ell$ string $S$ and a function $h=(\pi_{a_1},\ldots,\pi_{a_m}) \in \Projections^m$, we define $h(S)$ as $S[a_{1}] S[a_{2}] \cdots S[a_{m}]$.

Let $p_1 = 1 - k / \ell$ and $p_2 = 1 - (1+\eps) k / \ell$. We assume that $(1+\eps) k< \ell$ in order to guarantee $p_1>p_2>0$; the problem is trivial if $(1+\eps)k \ge \ell$. 
Further, let $m = \ceil { \log_{p_2}{\tfrac{1}{n}} }$.

We choose a set $\Hashes$ of $L = \Theta(n^{1/(1+\eps)})$ hash functions in $\Projections^m$ uniformly at random. Let $\Collisions^{\Hashes}_{\ell}$ be the mutliset of all collisions of length-$\ell$ substrings of $X$ and $Y$ under the functions from $\Hashes$, i.e. $\Collisions^{\Hashes}_{\ell} = \{(X[i,i+\ell-1], Y[j,j+\ell-1],h) : \varphi(h(X[i,i+\ell-1])) = \varphi(h(Y[j,j+\ell-1])), 1 \le i \le |X|-\ell, 1 \le j \le |Y|-\ell\}$. 

We will perform two tests. The first test chooses an arbitrary subset $\Collisions'\subseteq \Collisions^{\Hashes}_{\ell}$ of size $|\Collisions'|=\min\{4nL, |\Collisions^{\Hashes}_{\ell}|\}$ and, for each collision $(S_1, S_2, h)\in \Collisions'$, computes $\norm{\sk_\eps(S_1) - \sk_\eps (S_2)}^2$. If this value is at most $(1+\eps) k$, then the algorithm returns YES and the pair $(S_1, S_2)$ as a witness. The second test chooses a collision $(S_1, S_2, h) \in \Collisions^{\Hashes}_{\ell}$ uniformly at random and computes the Hamming distance between $S_1$ and $S_2$ character by character in $\Oh(\ell) = \Oh(n)$ time. If the Hamming distance is at most $(1+\eps)k$, the algorithm returns YES and the witness pair $(S_1, S_2)$. Otherwise, the algorithm returns NO. See Algorithm~\ref{alg:LSH}.

\begin{algorithm}[ht]
\caption{\kApproxLCS (decision variant)}
\begin{algorithmic}[1]
\State Choose a set $\Hashes$ of $L$ functions from $\Projections^m$ uniformly at random
\State $\Collisions^{\Hashes}_{\ell}\!=\!\{(S_1,S_2, h) : S_1, S_2\text{---length-$\ell$ substrings of $X,Y$ resp. and }\varphi(h(S_1)) = \varphi(h(S_2))\}$\label{ln:hashes}

\State Choose an arbitrary subset $\Collisions' \subseteq \Collisions^{\Hashes}_{\ell}$ of size $\min\{4nL,|\Collisions^{\Hashes}_{\ell}|\}$
\State Compute $\sk_\eps(\cdot)$ sketches for all length-$\ell$ substrings of $X, Y$  \label{ln:sketches}
\For {$(S_1, S_2, h) \in \Collisions'$} \label{ln:test1}
	\If {$\norm{\sk_\eps (S_1) - \sk_\eps (S_2)}^2 \le (1+\eps) k$} \label{ln:Hamming_dist}\Return $(\text{YES},(S_1,S_2))$
	\EndIf
\EndFor

\State Draw a collision $(S_1, S_2, h) \in \Collisions^{\Hashes}_{\ell}$ uniformly at random \label{ln:test2}
\If {$\HD (S_1, S_2) \le (1+\eps) k$} \label{ln:test3}\Return $(\text{YES},(S_1,S_2))$
\EndIf\\
\Return NO
\end{algorithmic}
\label{alg:LSH}
\end{algorithm}

We must explain how we compute $\Collisions^{\Hashes}_{\ell}$ and choose the collisions that we test. We consider each hash function $h \in \Hashes$ in turn. Let $h = (\pi_{a_1},\ldots,\pi_{a_m})$. Recall that for a string $S$ of length~$\ell$ we define $h(S)$ as $S[a_{1}] S[a_{2}] \cdots S[a_{m}]$. Consequently, $\varphi(h(S)) =  (\sum_{i=1}^m r^{i-1} S[a_{i}]) \bmod q$. We create a vector $U$ of length $\ell$ where each entry is initialised with $0$. For each $i$, we add $r^{i-1} \bmod q$ to the $a_{i}$-th entry of $U$. Finally, we run the FFT algorithm~\cite{FischerPaterson} for $U$ and $X, Y$ in the field $\mathbb{Z}_q$, and sort the resulting values. We obtain a list of sorted values that we can use to generate the collisions. Namely, consider some fixed value $z$. Assume that there are $x$ substrings of $X$ and $y$ substrings of $Y$ of length $\ell$ such that the fingerprint of their projection is equal to $z$. The value $z$ then gives $xy$ collisions, and we can generate each one of them in constant time. This explains how to choose the subset $\Collisions'$ in $\Oh(nL \log n)$ time.

To draw a collision from $\Collisions^{\Hashes}_{\ell}$ uniformly at random, we could simply compute the total number of collisions across all functions $h \in \Hashes$, draw a number in $[1,  |\Collisions^{\Hashes}_{\ell}|]$, and generate the corresponding collision. However, this would require to generate the collisions twice. Instead, we use the weighted reservoir sampling algorithm~\cite{reservoir}. We divide all collisions into subsets according to the values of fingerprints. We assume that the weighted reservoir sampling algorithm receives the fingerprint values one-by-one, as well as the number of corresponding collisions. At all times, the algorithm maintains a ``reservoir'' containing one fingerprint value and a random collision corresponding to this value. When a new value $z$ with $xy$ collisions arrives, the algorithm replaces the value in the reservoir with $z$ and a random collision with some probability. Note that to select a random collision it suffices to choose a pair from $[1,x] \times [1,y]$ uniformly at random. It is guaranteed that if for a value $z$ we have $xy$ collisions, the algorithm will select $z$ with probability $xy/|\Collisions^{\Hashes}_{\ell}|$. Consequently, after processing all values, the reservoir will contain a collision chosen from $\Collisions^{\Hashes}_{\ell}$ uniformly at random.

\begin{lemma}\label{lm:complexity}
Algorithm~\ref{alg:LSH} uses $\Oh(n^{1+1/(1+\eps)} \log n + \sigma n \log^2 n)$ time and $\Oh(n \log n)$ space. 
\end{lemma}
\begin{proof}
Computing the sketches (Line~\ref{ln:sketches}) takes $\Oh(\sigma n \log^2 n)$ time and $\Oh(n \log n)$ space. Computing the collisions and choosing the collisions to test takes $\Oh(n^{1+1/(1+\eps)} \log n)$ time and $\Oh(n)$ space in total. Testing $\min \{4nL, |\Collisions^{\Hashes}_{\ell}|\}$ collisions (Line~\ref{ln:test1}) takes $\Oh(n^{1+1/(1+\eps)} \log n)$ time and constant space. Computing the Hamming distance for a random collision (Line~\ref{ln:test3}) takes $\Oh(\ell) = \Oh(n)$ time and constant space.
\end{proof}

\begin{lemma}\label{lm:hash_function_exists}
Let $S_1$ and $S_2$ be two length-$\ell$ substrings of $X$ and $Y$, resp., with $\HD(S_1, S_2) \le k$. 
If $L=\Theta(n^{1/(1+\eps)})$ is large enough, then, with probability at least $3/4$, there exists a function $h\in \Hashes$ such that $h(S_1)=h(S_2)$.
\end{lemma}
\begin{proof}
Consider a function $h=(\pi_{a_1},\ldots,\pi_{a_m})$ drawn from $\Projections^m$ uniformly at random. The probability of $h(S_1)=h(S_2)$ is at least $p_1^m$.
Due to $p_1 \le 1$, we have \[p_1^m = p_1^{\ceil{\log_{p_2}\frac1n}} \ge p_1^{1+\log_{p_2}\frac1n}=p_1\cdot n^{-\frac{\log p_1}{\log p_2}}.\]
Moreover, $p_1 = 1-\frac{k}{\ell}$ and $(1+\eps)k < \ell$ yield $p_1 > 1-\frac{1}{1+\eps}=\frac{\eps}{1+\eps}$,
whereas Bernoulli's inequality implies $p_2 = 1-(1+\eps)\frac{k}{\ell} \le (1-\frac{k}{\ell})^{1+\eps}=p_1^{1+\eps}$,
i.e., $\log p_2  \le (1+\eps)\log p_1$.
Therefore, \[p_1^m \ge p_1\cdot n^{-\frac{\log p_1}{\log p_2}}\ge\tfrac{\eps}{1+\eps}\cdot n^{-\frac{1}{1+\eps}}.\]
Hence, we can choose the constant in $L=|\Hashes|$ so that the claim of the lemma holds.
\end{proof}

\begin{lemma}\label{lm:bad_collisions}
If $|\Collisions^{\Hashes}_{\ell}| > 4 nL$ and $(S_1,S_2, h)$ is a uniformly random element of $\Collisions^{\Hashes}_{\ell}$, then $\Prob[\HD(S_1,S_2) \ge (1+\eps) k] \le \frac12$.
\end{lemma}
\begin{proof}
Consider length-$\ell$ substrings $S_1, S_2$ of $X, Y$, respectively, such that $\HD(S_1, S_2) \ge (1+\eps)k$, and a hash function $h$. Let us bound the probability of $(S_1,S_2, h) \in \Collisions^{\Hashes}_{\ell}$. There two possible cases: either $h(S_1) \neq h(S_2)$ but $\varphi(h(S_1)) = \varphi(h(S_2))$, or $h(S_1) = h(S_2)$. The probability of the first event is bounded by the collision probability of Karp--Rabin fingerprints, which is at most $1/n$. Let us now bound the probability of the second event. Since $\HD (S_1,S_2)\ge (1+\eps)k$, we have $\Prob [h (S_1) = h (S_2)] \le p_2^{m} \le 1/n$,  
where the last inequality follows from the definition of $m$. Therefore, the probability that for some function $h \in \Hashes$ we have $\varphi(h(S_1)) = \varphi(h(S_2))$ is at most $2/n$. 

In total, we have $n^2 |\Hashes|$ possible triples $(S_1, S_2 ,h)$ so by linearity of expectation, we conclude that the expected number of such triples is at most $\frac{2}{n} n^2 L =2n L$. Therefore the probability to hit a triple $(S_1, S_2, h)$ such that $\HD(S_1, S_2) \ge (1+\eps)k$ when drawing from $\Collisions^{\Hashes}_{\ell}$ uniformly at random is at most $2nL / |\Collisions^{\Hashes}_{\ell}| \le 2nL / 4nL = 1/2$.
\end{proof}

Below, we combine the previous results to prove that, with constant probability, Algorithm~\ref{alg:LSH} correctly solves
the decision variant of the \kApproxLCS problem.
Note that we can reduce the error probability to an arbitrarily small constant $\delta>0$: it suffices to repeat the algorithm a constant number of times. 

\begin{corollary}
With non-zero constant probability, Algorithm~\ref{alg:LSH} solves the decision variant of \kApproxLCS correctly.
\end{corollary}
\begin{proof}
Suppose first that $\ell \le \lcsk(X,Y)$, which means that there are two length-$\ell$ substrings $S_1, S_2$ of $X, Y$ such that $\HD(S_1, S_2) \le k$. By Lemma~\ref{lm:hash_function_exists}, with probability at least $3/4$, there exists a function $h\in \Hashes$ such that $h(S_1)=h(S_2)$.
In other words, $(S_1, S_2, h) \in \Collisions^{\Hashes}_{\ell}$ with probability at least $\frac34$. If $|\Collisions^{\Hashes}_{\ell}| < 4nL$, we will find this triple and it will pass the test with probability at least $1-n^{-6}$.
If $|\Collisions^{\Hashes}_{\ell}| \ge 4nL$, then by Lemma~\ref{lm:bad_collisions} the Hamming distance between $S_1, S_2$, where $(S_1, S_2, h)$ was drawn from $\Collisions^{\Hashes}_{\ell}$ uniformly at random, is at most $(1+\eps)k$ with probability $\ge 1/2$, and therefore this pair will pass the test with probability $\ge 1/2$. It follows that in this case the algorithm outputs YES with constant probability.

Suppose now that $\ell > \lcske(X,Y)$. In this case, the Hamming distance between any pair of length-$\ell$ substrings of $X$ and $Y$ is at least $(1+\eps)k$, so none of them will ever pass the second test and none of them will pass the first test with constant probability.
\end{proof}

We now improve the space of the algorithm to linear. Note that the only reason why we needed $\Oh(n \log n)$ space is that we precompute and store the sketches for the Hamming distance. Below we explain how to overcome this technicality.

First, we do not precompute the sketches. Second, we process the collisions in $\Collisions'$ in batches of size $n$. Consider one of the batches, $\mathcal{B}$. For each collision $(S_1,S_2, h) \in \mathcal{B}$ we must compute $\norm{\sk_\eps(S_1) - \sk_\eps(S_2)}^2$. 
We initialize a counter for every collision, setting it to zero initially. The number of rounds in the algorithm will be equal to the length of the sketches, and, in round $i$, the counter for a collision $(S_1, S_2, h) \in \mathcal{B}$ will contain the squared $L_2$ distance between the length-$i$ prefixes of $\sk_\eps(S_1)$ and $\sk_\eps(S_2)$. In more detail, let $\mathcal{S}$ be the set of all substrings of $X, Y$ that participate in the collisions in $\mathcal{B}$. Recall that all these substrings have length $\ell$. At round $i$, we compute the $i$-th coordinate of the sketches of the substrings in $\mathcal{S}$. By definition, the $i$-th coordinate is the dot product of the $i$-th row of $c \cdot M$, where $c$ and $M$ are as in Corollary~\ref{cor:dim_reduction}, and a substring encoded using $\mu$. Hence, we can compute the coordinate using the FFT algorithm~\cite{FischerPaterson} in $\Oh(\sigma n \log n)$ time and $\Oh(n)$ space. When we have the coordinate $i$ computed, we update the counters for the collisions and repeat.

At any time, the algorithm uses $\Oh(n)$ space.
Compared to the time consumption proven in Lemma~\ref{lm:complexity}, the algorithm spends an additional $\Oh(\sigma n^{1+1/(1+\eps)} \log^2 n)$ time for computing the coordinates of the sketches.
Therefore, in total the algorithm uses $\Oh(\sigma n^{1+1/(1+\eps)} \log^2 n) = \Oh(n^{1+1/(1+\eps)} \log^2 n)$ time and $\Oh(n)$ space. 
For constant-size alphabets, this completes the proof of Theorem~\ref{th:LCS'k-decision}. For alphabets of arbitrary size, we replace the sketches from Section~\ref{sec:prelim} with the sketches defined in~\cite{DBLP:journals/algorithmica/KociumakaRS19} to achieve the desired complexity.
We note that we could use the sketches~\cite{DBLP:journals/algorithmica/KociumakaRS19} for small-size alphabets as well, but their lengths hide a large constant.

\section{Experiments}\label{sec:implem}
We now present results of experimental evaluation of the second solution presented in Theorem~\ref{th:klcs_upper}.

\textit{Methodology and test environment.} The baselines and our solution are written in C++11 and compiled with optimizations using gcc 7.4.0. The experimental results were generated on an Intel Xeon E5-2630 CPU using 128 GiB RAM. To ensure the reproducibility of our results, our complete experimental setup, including data files, is available at \url{https://github.com/fnareoh/LCS\_Approx\_k\_mis}.

\textit{Baseline.} The only other solution to the \kApproxLCS problem was presented in~\cite{DBLP:journals/algorithmica/KociumakaRS19}, however, it has a worse complexity and is likely to be unpractical because it uses a very complex class of hash functions. We therefore chose to compare our algorithm against algorithms for the \kLCS problem. To the best of our knowledge, none of the existing algorithms has been implemented. We implemented the solution to \kLCS by Flouri et al., which we refer to as FGKU~\cite{DBLP:journals/ipl/FlouriGKU15}. (The other algorithms seem to be too complex to be efficient in practice.) The main idea of the algorithm of Flouri et al.\ is that if we know that the longest common substring with $k$ mismatches is obtained by a substring of $X$ that starts at a position $p$ and a substring of $Y$ that starts at a position $p+i$, then we can find it by scanning $X$ and $Y[i,|Y|]$ in linear time; see Algorithm~\ref{alg:FGKU} for details.

\begin{algorithm}[ht]
\caption{FGKU algorithm}
\begin{algorithmic}[1]
\State $n \gets  |X|$, $m \gets  |Y|$
\State $l \gets  0$, $r_1 \gets  0$, $r_2 \gets  0$
\For {$d \gets-m+1$ to $n-1$} 
	\State $i \gets  \max(-d,0)+d$, $j \gets  \max(-d,0)$
	\State $Q \gets  \emptyset$, $s \gets  0$, $p \gets  0$ 
	\While {$p \leq \min(n-i,m-j)-1$}
	\If {$X[i+p] \neq Y[j+p]$}
		\If {$|Q| = k$}
			\State $s \gets  \min Q + 1$
			\State {\footnotesize DEQUEUE}($Q$)
		\EndIf
		\State {\footnotesize ENQUEUE}($Q,p$)
	\EndIf
	\State $p \gets p+1$
	\If{$p-s > l$}
		\State $l \gets p-s$, $r_1 \gets i+s$, $r_2 \gets j+s$
	\EndIf
	\EndWhile
\EndFor
\end{algorithmic}
\label{alg:FGKU}
\end{algorithm}

\textit{Details of implementation.}
We made several adjustments to the theoretical algorithm we described. First, we use the fact that $A = \mathrm{LCS}(X,Y)+k \le \lcsk(X,Y) \le B = (k+1)\cdot\mathrm{LCS}(X,Y)+k$ to bound the interval in the \twentyquestions game. We also treated the number of questions in the  \twentyquestions game and~$L$, the size of the set of hash functions $\Hashes$, as parameters that trade time for accuracy, and put the number of questions to $2 \log (B-A)$ in the \twentyquestions game and $L = n^{1/(1+\eps)}/16$. In Line~\ref{ln:Hamming_dist} of Algorithm~\ref{alg:LSH}, we used sketches to estimate the Hamming distance. In practice, we computed the Hamming distance via character-by-character comparison when $\ell$ is small compared to $k$ and via kangaroo jumps otherwise~\cite{10.1145/8307.8309}. Also, when the length $\ell$ in Algorithm~\ref{alg:LSH} is smaller than $2 \log n$, we compute the hash values of the $\ell$-length substrings of $S_1$ and $S_2$ naively, instead of using the FFT algorithm~\cite{FischerPaterson}.

\begin{figure}[ht!]
\centering
    \begin{subfigure}{0.5\textwidth}
        \centering
        \captionsetup{justification=centering}
        \includegraphics[scale=0.45]{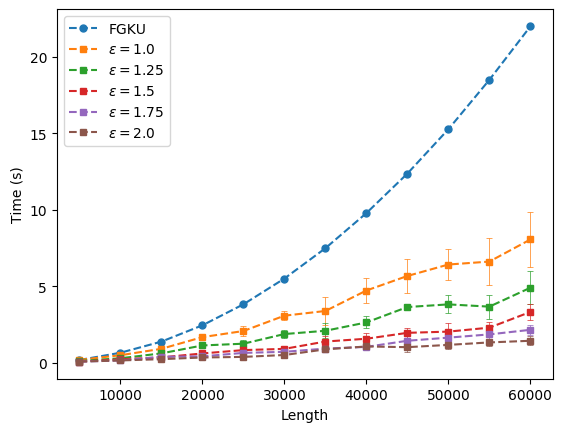}
        \caption{Random, $k = 10$}
    \end{subfigure}%
    \begin{subfigure}{0.5\textwidth}
        \centering
        \captionsetup{justification=centering}
        \includegraphics[scale=0.45]{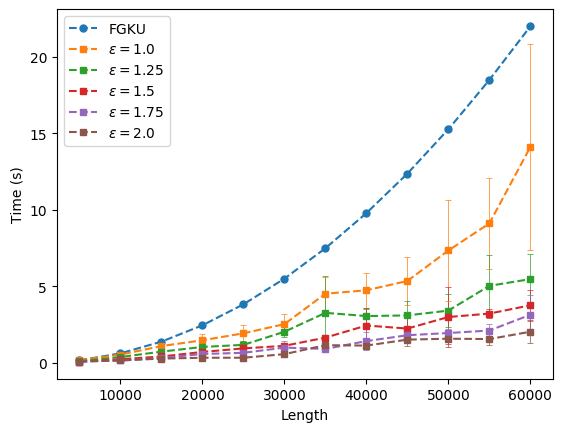}
        \caption{E. coli, $k = 10$}
    \end{subfigure}
\caption{Comparison of the FGKU algorithm versus our algorithm for $k = 10$ and different values of $\eps$. Large standard deviation for length $60000$ is caused by an outlier with very long longest common substring with $k$ mismatches.}
\label{fig:runtime_10}
\end{figure}

\begin{figure}[ht!]
\centering
    \begin{subfigure}{0.5\textwidth}
        \centering
        \captionsetup{justification=centering}
        \includegraphics[scale=0.45]{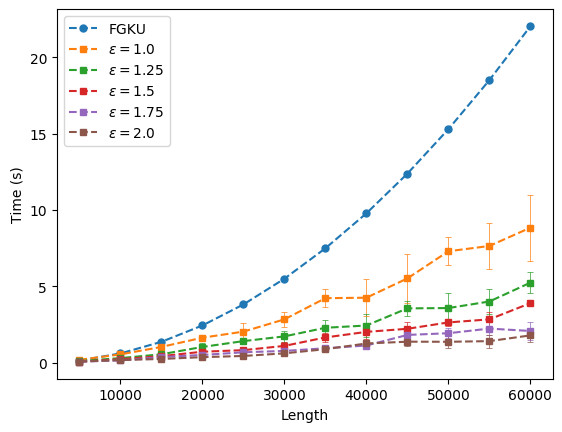}
        \caption{Random, $k = 25$}
    \end{subfigure}%
    \begin{subfigure}{0.5\textwidth}
        \centering
        \captionsetup{justification=centering}
        \includegraphics[scale=0.45]{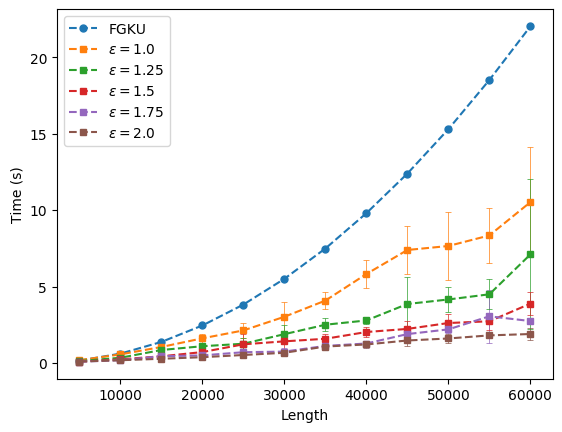}
        \caption{E. coli, $k = 25$}
    \end{subfigure}     
\caption{Comparison of the FGKU algorithm versus our algorithm for $k = 25$ and different values of $\eps$.}
\label{fig:runtime_25}
\end{figure}
 
\begin{figure}[ht!]
\centering   
    \begin{subfigure}{0.5\textwidth}
        \centering
        \captionsetup{justification=centering}
        \includegraphics[scale=0.45]{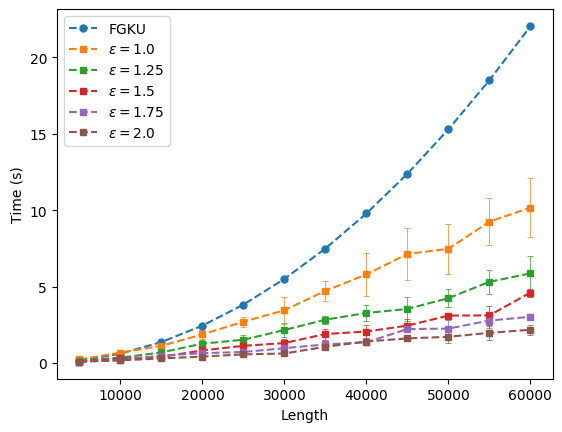}
        \caption{Random, $k = 50$}
    \end{subfigure}%
    \begin{subfigure}{0.5\textwidth}
        \centering
        \captionsetup{justification=centering}
        \includegraphics[scale=0.45]{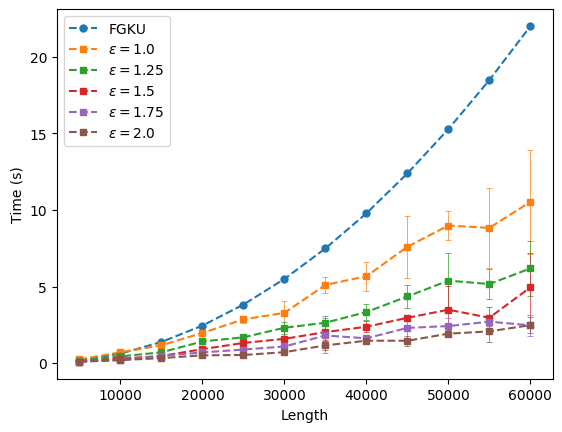}
        \caption{E. coli, $k = 50$}
    \end{subfigure}        
\caption{Comparison of the FGKU algorithm versus our algorithm for $k = 50$ and different values of $\eps$.}
\label{fig:runtime_50}
\end{figure}

\textit{Data sets and results.}
We considered $k \in \{10, 25, 50\}$ and $\eps \in \{1.0, 1.25, 1.5, 1.75, 2.0\}$. We tested the algorithms on pairs of random strings (each character is selected independently and uniformly from a four-character alphabet $\{A, T, G, C\}$) and on pairs of strings extracted at random from the E. coli genome. The lengths of the strings in each pair are equal and vary from $0$ to $60000$ with a step of $5000$. All timings reported are averaged over ten runs. Figures~\ref{fig:runtime_10}-~\ref{fig:runtime_50} show the results for $k = 10, 25, 50$. We note that for $\eps = 1$ and $k = 10, 25$, the standard deviation of the running time on the E. coli data set is quite large, which is probably caused by our choice of the method to compute the Hamming distance between substrings, but for all other parameter combinations it is within the standard range. We can see that the time decreases when $\eps$ grows, which is coherent with the theoretical complexity.

As for the accuracy, note that our algorithm cannot return a pair of strings at Hamming distance more than $(1+\eps) k$, and so the only risk is returning strings which are too short. Consequently, we measured the accuracy of our implementation by the ratio of the length $\lcsak(X, Y)$ returned by our algorithm divided by $\lcsk(X, Y)$ computed by the dynamic programming. We estimate $r_{\min}(\eps, k) = \min_{X,Y}(\lcsak(X,Y)/\lcsk(X,Y))$ and $r_{\max}(\eps, k) = \max_{X,Y}(\lcsak(X,Y)/\lcsk(X,Y))$
by computing $\lcsak$ and $\lcsk$ for $10$ pairs of strings for each length from $5000$ to $60000$ with step of $5000$, as well as the error rate, i.e. the percentage of experiments where $\lcsak(X,Y)$ is shorter than $\lcsk(X,Y)$ (see Table~\ref{tb:eps}). Not surprisingly, $r_{\min}$ and~$r_{\max}$ grow as $k$ and $\eps$ grow, while the error rate drops. Even though there is no theoretical upper bound on $r_{\max}$, the latter is at most $2.24$ at all times. We also note that even in the cases when the error rate is non-negligible, $\lcsak \ge 0.86 \cdot \lcsk$, in other words, our algorithm returns a reasonable approximation of $\lcsk$.

\newcolumntype{?}{!{\vrule width 1pt}}
\newcommand{\err}{\mathrm{err}}
\begin{center}
\begin{table}[ht!]
\center
\begin{tabular}{| c ? c | c | c | c| c| c ? c | c | c | c | c | c |}
\hline
 & \multicolumn{6}{c?}{Random} & \multicolumn{6}{c|}{E. coli} \\ 
\hline
 & \multicolumn{2}{c|}{$k = 10$} & \multicolumn{2}{|c|}{$k = 25$} & \multicolumn{2}{|c?}{$k = 50$} & \multicolumn{2}{c|}{$k = 10$} & \multicolumn{2}{c|}{$k = 25$} & \multicolumn{2}{c|}{$k = 50$} \\ 
\hline
\hline

\multirow{ 2}{*}{$\eps = 1.0$} & 0.95 & 1.41 & 1.12 & 1.46 & 1.27 & 1.54 & 0.89 & 1.34 & 0.94 & 1.48 & 0.97 & 1.59\\ 
\cline{2-13}
& \multicolumn{2}{c|}{$\err = 3\%$}  & \multicolumn{2}{c|}{$\err = 0\%$}   & \multicolumn{2}{c?}{$\err = 0\%$}   & \multicolumn{2}{c|}{$\err = 33\%$}   & \multicolumn{2}{c|}{$\err = 13\%$}  & \multicolumn{2}{c|}{$\err = 3\%$}\\ 
\hline
\multirow{ 2}{*}{$\eps = 1.25$}  & 0.97 & 1.47 & 1.15 & 1.63 & 1.44 & 1.78 & 0.88  & 1.48 & 0.98 & 1.56 & 0.99 & 1.73\\ 
\cline{2-13}
& \multicolumn{2}{c|}{$\err = 1\%$}  & \multicolumn{2}{c|}{$\err = 0\%$}   & \multicolumn{2}{c?}{$\err = 0\%$}   & \multicolumn{2}{c|}{$\err = 28\%$}   & \multicolumn{2}{c|}{$\err = 5\%$}  & \multicolumn{2}{c|}{$\err = 3\%$}\\ 
\hline
\multirow{ 2}{*}{$\eps = 1.5$}  & 1.05 & 1.57 & 1.37 & 1.76 & 1.55 & 1.91  & 0.88 & 1.45 & 0.96 & 1.67 & 0.99 & 1.89\\ 
\cline{2-13}
& \multicolumn{2}{c|}{$\err = 0\%$}  & \multicolumn{2}{c|}{$\err = 0\%$}   & \multicolumn{2}{c?}{$\err = 0\%$}   & \multicolumn{2}{c|}{$\err = 17\%$}   & \multicolumn{2}{c|}{$\err = 3\%$}  & \multicolumn{2}{c|}{$\err = 3\%$}\\  
\hline
\multirow{ 2}{*}{$\eps = 1.75$} & 1.02 & 1.69 & 1.46 & 1.86 & 1.72 & 2.12 & 0.88 & 1.58 & 0.95 & 1.84 & 1.02 & 2.15\\ 
\cline{2-13}
& \multicolumn{2}{c|}{$\err = 0\%$}  & \multicolumn{2}{c|}{$\err = 0\%$}   & \multicolumn{2}{c?}{$\err = 0\%$}   & \multicolumn{2}{c|}{$\err = 17\%$}   & \multicolumn{2}{c|}{$\err = 2\%$}  & \multicolumn{2}{c|}{$\err = 0\%$}\\ 
\hline
\multirow{ 2}{*}{$\eps = 2.0$}  & 1.10 & 1.72 & 1.59 & 2.00 & 1.89 & 2.24 & 0.91  & 1.77 & 1.01 & 2.10 & 1.00 & 2.19\\ 
\cline{2-13}
& \multicolumn{2}{c|}{$\err = 0\%$}  & \multicolumn{2}{c|}{$\err = 0\%$}   & \multicolumn{2}{c?}{$\err = 0\%$}   & \multicolumn{2}{c|}{$\err = 9\%$}   & \multicolumn{2}{c|}{$\err = 0\%$}  & \multicolumn{2}{c|}{$\err = 1\%$}\\ 
\hline
\end{tabular} 
\caption{Accuracy of the \kApproxLCS algorithm. For each $k$ and $\eps$, we show $r_{\min}(\eps, k)$, $r_{\max}(\eps, k)$, as well as the error rate.}
\label{tb:eps}
\end{table}
\end{center}

\section{Proof of Fact~\ref{lm:klcs_lower}}\label{sec:hardness}
We now show the lower bound of Fact~\ref{lm:klcs_lower} by reduction from the \Bichromatic problem.

\begin{problem}[\Bichromatic]\label{pr:bichromatic}
Given a constant $\gamma > 0$ and two sets of binary strings $(U_i)_{i \in [1,N]}$ and $(V_j)_{j \in [1,N]}$, each of length $d = \Oh(\log N)$, if the smallest Hamming distance between a pair $(U_i,V_j)_{i,j \in [1,N]}$ is $h$, we must output (possibly another) pair of binary strings with Hamming distance in $[h, (1+\gamma)h]$.
\end{problem}

Rubinstein~\cite{DBLP:journals/corr/abs-1803-00904} proved that for every $\delta$ there exists $\gamma = \gamma (\delta)$ such that any randomised algorithm that solves \Bichromatic correctly with constant probability requires $\Oh(N^{2-\delta})$ time assuming SETH:

\begin{hypothesis}[SETH]
For every $\delta > 0$, there exists an integer $q$ such that SAT on $q$-CNF formulas with~$m$ clauses and  $n$ variables cannot be solved in $m^{O(1)} 2^{(1-\delta) n}$ time even by a randomised algorithm\footnote{Impagliazzo, Paturi, and Zane~\cite{DBLP:journals/jcss/ImpagliazzoPZ01} stated the hypothesis for deterministic algorithms only, but nowadays it is common to extend SETH to allow randomisation. If we condition on the classic version of the hypothesis, we will obtain a lower bound for deterministic algorithms. See~\cite{SETH_survey} for more discussion.}.
\end{hypothesis}

We show the lower bound by reducing from \Bichromatic to a polylogarithmic number of instances of \kApproxLCS. 

\subparagraph*{Quaternary alphabet. } 
First, we show a proof for an alphabet of four characters $\{0,1,a,b\}$, and then explain how to adapt the proof for the binary alphabet. We assume that $U_i$, $V_j$ are over the alphabet $\{0,1\}$. Let us introduce $H=(a^d b)^{d+1}$, and construct $X = H U_1 H U_2 H \ldots H U_N H$ and $Y = H V_1 H V_2 H \ldots H V_N H$.

\begin{observation}\label{obs:bichro_lcs}
If there exist $i,j \in [1,N] $ such that $\HD(U_i,V_j) \leq k$, then  $\lcsk(X,Y) \geq 2(d+1)^2 + d$.
\end{observation}
\begin{proof}
If $\HD(U_i,V_j) \leq k$ for some $i, j$, then $\HD(HU_iH,HV_jH)  \leq k$ and $\lcsk(X,Y) \geq |H U_i H| = 2(d+1)^2 + d$.
\end{proof}

\begin{lemma} \label{lm:bichro_lcs}
For a given $k \leq d/(1+\eps)$, if the algorithm for \kApproxLCS outputs a substring $S_1$ of $X$ and a substring $S_2$ of $Y$ of length $\geq 2(d+1)^2 + d$, then there exist $i,j \in [1,N] $ such that $\HD(U_i,V_j) \leq (1+\eps)k$.
\end{lemma}

\begin{proof}
By the assumption of the lemma, $|S_1| = |S_2| \geq 2(d+1)^2 + d$. $S_2$ contains either $HV_j$ or $V_jH$ for some $j$. W.l.o.g., we can assume that $S_2$ contains a copy of $H$ followed by $V_j$ for some $j$. Let us consider the substring $S$ of $X$ the copy of $H$ is aligned with. Below we will prove that $S = H$ and since it is followed by~$U_i$ for some $i$, it will imply that $\HD(HU_iH,HV_jH)  \leq (1+\eps)k$.

\begin{figure}
\definecolor{mygray}{gray}{0.6}
\definecolor{lessgray}{gray}{0.85}

\begin{center}
{\footnotesize
\begin{tikzpicture}[scale=0.6]
  \filldraw[lessgray] (0,0) rectangle (4,1);
  \filldraw[mygray] (4,0) rectangle (5,1);
  \filldraw[lessgray] (5,0) rectangle (9,1);
  \filldraw[mygray] (9,0) rectangle (10,1);
  \filldraw[lessgray] (10,0) rectangle (14,1);
  
  \draw (-1.5,0.5) node {$X$};
  \draw (-0.5,0.5) node {$\ldots$};
  \draw (14.5,0.5) node {$\ldots$};
  
  \draw (2,0.5) node {\small{$H$}};
  \draw (4.5,0.5) node {\small{$U_{i-1}$}};
  \draw (7,0.5) node {\small{$H$}};
  \draw (9.5,0.5) node {\small{$U_i$}};
  \draw (12,0.5) node {\small{$H$}};
  
  
  \filldraw[lessgray, xshift = 2.5 cm, yshift= -2.5 cm] (0,0) rectangle (4,1);
  \filldraw[mygray, xshift = 2.5 cm, yshift= -2.5 cm] (4,0) rectangle (5,1);
  \filldraw[lessgray, xshift = 2.5 cm, yshift= -2.5 cm] (5,0) rectangle (9,1);
  \filldraw[mygray, xshift = 2.5 cm, yshift= -2.5 cm] (9,0) rectangle (10,1);
  \filldraw[lessgray, xshift = 2.5 cm, yshift= -2.5 cm] (10,0) rectangle (14,1);
  
  \draw[xshift = 2.5 cm, yshift= -2.5 cm] (-1.5,0.5) node {$Y$};
  \draw[ xshift = 2.5 cm, yshift= -2.5 cm] (-0.5,0.5) node {$\ldots$};
  \draw[ xshift = 2.5 cm, yshift= -2.5 cm] (14.5,0.5) node {$\ldots$};

  \draw[xshift = 2.5cm, yshift= -2.5 cm] (4.5,0.5) node {\small{$V_{j-1}$}};
  \draw[xshift = 2.5cm, yshift= -2.5 cm] (7,0.5) node {\small{$H$}};
  \draw[xshift = 2.5cm, yshift= -2.5 cm] (9.5,0.5) node {\small{$V_j$}};
  \draw[xshift = 2.5cm, yshift= -2.5 cm] (12,0.5) node {\small{$H$}};
  
  \draw [decorate,decoration={brace,amplitude=0.25cm}](4.5,1) -- (13.5,1) node [black,midway, yshift=0.5cm]{$S_1$};
  \draw [decorate,decoration={brace,amplitude=0.25cm, mirror}, yshift= -3.5 cm](4.5,1) -- (13.5,1) node [black,midway, yshift=-0.5cm]{$S_2$};
  \draw [decorate,decoration={brace,amplitude=0.3cm,mirror}](7.5,0) -- (11.5,0) node [black,midway, yshift=-0.5cm]{$S$};
  \draw [dashed] (7.5,1) -- (7.5,-2.5);
  \draw [dashed] (11.5,1) -- (11.5,-2.5);
  \draw[<->] (5,-0.25) -- (7.5,-0.25);
  \draw[black] (6.25,-0.5) node {$s$};

\end{tikzpicture}

}
\end{center}
\caption{Substrings $S_1$ and $S_2$ of $X$ and $Y$, respectively, substring $S$ aligned with a copy of $H$ in $S_2$, and the shift $s$.}
\end{figure}
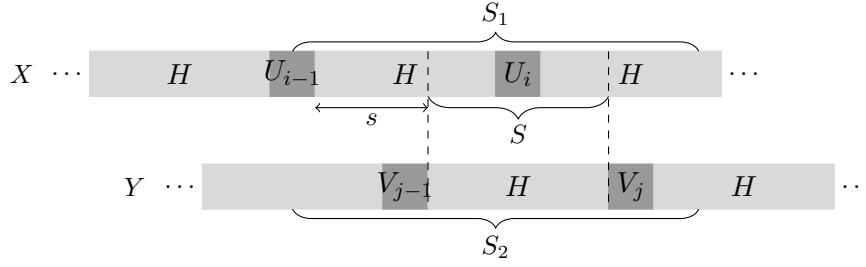

Assume $S \neq H$, and let $0 < s < (d+1)^2 + d$ be the distance between the starting positions of $S$ and the nearest copy of $H$ from the left. If $s < d+1$ or $(d+1)^2 < s$, then all occurrences of $b$ in $H$ create a mismatch. There are $d+1 > (1+\eps)k$ of them, a contradiction. If $d+1 \leq s \le (d+1)^2$, then $H$ covers $U_i$ creating $d$ mismatches. Since $|U_i| = d$ and $|H|=(d+1)^2$, we will have at least one more mismatch from the alignment of the copy of $H$ in $Y$ and the copies of $H$ in $X$ that surround $U_i$. Therefore, in total there are at least $d+1 > (1+\eps)k$ mismatches, a contradiction. To conclude, both cases are impossible and hence $s = 0$. The lemma follows as explained above.
\end{proof}
With this lemma, we can now proceed to prove Fact~\ref{lm:klcs_lower}.
Define $\eps = \gamma/3$ and consider all $k = 1, (1+\eps), (1+\eps)^2, (1+\eps)^3, \ldots$ until $d/(1+\eps)$. For each $k$, we run $\log \log n$ independent instances of an algorithm for \kApproxLCS. Let $k_0$ be the smallest $k$ such that the longest common substring with approximately $k$ mismatches has length at least $2(d+1)^2 + d$.

By the definition of $k_0$, Observation~\ref{obs:bichro_lcs} and Lemma~\ref{lm:bichro_lcs}, there do not exist $i,j \in [1,N]$ such that $\HD(U_i,V_j) \leq k_0/(1+\eps)$ but there exist $i,j \in [1,N]$ such that $\HD(U_i,V_j) \leq k_0(1+\eps)$.
In the \Bichromatic problem, this translates to $ k_0/(1+\eps) < h \leq k_0(1+\eps)$, where $h$ is the minimal distance between all pairs $U_i,V_j$.
This is equivalent to
\[h \leq k_0(1+\eps) < h(1+\eps)^2 = h(1+\tfrac23\gamma + \tfrac19\gamma^2) \leq h(1+\gamma),\]
which means that the pair $(U_i,V_j)$ found by the algorithm for $k_0$ is a valid solution for \Bichromatic. It follows that for some $k$, the algorithm for \kApproxLCS must spend $\Omega(N^{2-\delta} / \log_{1+\eps} \log N)$ time. We have $n = |X| = |Y| = \Oh(d^2 N) = \Oh(N \log^2 N)$, which implies $N = \Omega(n / \log^2 n)$. Fact~\ref{lm:klcs_lower} follows.

\vspace*{-10pt}
\subparagraph*{Binary alphabet. } 
To prove the lower bound for the binary alphabet, it suffices to show an analogue of Observation~\ref{obs:bichro_lcs} and Lemma~\ref{lm:bichro_lcs}; the rest follows as above. To this end, we introduce a morphism $\mu$ defined as $\mu(0)= 0001000$, $\mu(1)= 0011000$, $\mu(a)= 1001000$ and $\mu(b)= 1101000$.

\begin{observation}\label{obs:11}
Distances $\HD(\mu(0),\mu(a))$, $\HD(\mu(0),\mu(b))$, $\HD(\mu(1),\mu(a))$, $\HD(\mu(1),\mu(b))$ are at least $1$, and $\HD(\mu(0),\mu(1))=1$.
\end{observation}

We redefine $X$ as $\mu(X)$ and $Y$ as $\mu(Y)$. We will show that for a given $k \leq d/(1+\eps)$, if there exist $i,j \in [1,N] $ such that $\HD(U_i,V_j) \leq k$, then $\lcsk(X,Y) \geq 14(d+1)^2 + 7d$, and if the algorithm for \kApproxLCS outputs outputs a pair of substrings of length $\geq 14(d+1)^2 + 7d$, then there exist $i,j \in [1,N] $ such that $\HD(U_i,V_j) \leq (1+\eps)k$. Since  $\HD(\mu(0),\mu(1))=1$, the first part of the claim follows similar to Observation~\ref{obs:bichro_lcs}. To show the second part of the claim, we will need the following observation:

\begin{observation}\label{obs:1000}
Let $x,y,z \in \{0,1\}$. The string $1000$ has exactly two occurrences in $1000xyz1000$.
\end{observation}

Assume that the algorithm for \kApproxLCS outputs a substring $S_1$ of $X$ and a substring $S_2$ of $Y$.
We have $|S_1| = |S_2| \geq 14(d+1)^2 + 7d$. W.l.o.g., we can assume that $S_2$ contains a copy of $\mu(H)$ followed by $\mu(V_j)$ for some $j$. Let us consider the substring $S$ of $X$ the copy of $\mu(H)$ is aligned with. Below we will prove that $S = \mu(H)$ and since it is followed by $\mu(U_i)$ for some $i$, it will imply that $\HD(\mu(HU_iH),\mu(HV_jH))  \leq (1+\eps)k$ and consequently $\HD(HU_iH,HV_jH)  \leq (1+\eps)k$ by Observation~\ref{obs:11}.

We define the shift $s$ analogous of Lemma~\ref{lm:bichro_lcs}. We will prove that $s = 0 \pmod 7$, which implies that the image of any character in $\mu(H)$ is aligned with an image of some character in $S$. The claim will immediately follow analogous to Lemma~\ref{lm:bichro_lcs}.
Consider all $(d+1)^2$ occurrences of $1000$ in $\mu(H)$. We claim that they must be aligned with occurrences of $1000$ in $S$. If this is not the case, by Observation~\ref{obs:1000}, each occurrence of $1000$ accounts for at least one mismatch and therefore there are at least $(d+1)^2$ mismatches in total, a contradiction.

\bibliography{main}

\begin{thebibliography}{10}

\bibitem{DBLP:conf/soda/AbboudWY15}
Amir Abboud, Richard~Ryan Williams, and Huacheng Yu.
\newblock More applications of the polynomial method to algorithm design.
\newblock In {\em SODA'15}, pages 218--230, 2015.
\newblock \href {http://dx.doi.org/10.1137/1.9781611973730.17}
  {\path{doi:10.1137/1.9781611973730.17}}.

\bibitem{ACHLIOPTAS2003671}
Dimitris Achlioptas.
\newblock Database-friendly random projections: {J}ohnson-{L}indenstrauss with
  binary coins.
\newblock {\em Journal of Computer and System Sciences}, 66(4):671--687, 2003.
\newblock \href {http://dx.doi.org/10.1016/S0022-0000(03)00025-4}
  {\path{doi:10.1016/S0022-0000(03)00025-4}}.

\bibitem{agrawal2004primes}
Manindra Agrawal, Neeraj Kayal, and Nitin Saxena.
\newblock {PRIMES} is in {P}.
\newblock {\em Annals of Mathematics}, 160(2):781--793, 2004.
\newblock \href {http://dx.doi.org/10.4007/annals.2004.160.781}
  {\path{doi:10.4007/annals.2004.160.781}}.

\bibitem{substringNN}
Alexandr Andoni and Piotr Indyk.
\newblock Efficient algorithms for substring near neighbor problem.
\newblock In {\em SODA'06}, pages 1203--1212, 2006.
\newblock \href {http://dx.doi.org/10.1145/1109557.1109690}
  {\path{doi:10.1145/1109557.1109690}}.

\bibitem{DBLP:conf/stoc/AndoniR15}
Alexandr Andoni and Ilya Razenshteyn.
\newblock Optimal data-dependent hashing for approximate near neighbors.
\newblock In {\em STOC'15}, pages 793--801, 2015.
\newblock \href {http://dx.doi.org/10.1145/2746539.2746553}
  {\path{doi:10.1145/2746539.2746553}}.

\bibitem{DBLP:journals/poit/BabenkoS11}
Maxim Babenko and Tatiana Starikovskaya.
\newblock Computing the longest common substring with one mismatch.
\newblock {\em Problems of Information Transmission}, 47(1):28--33, 2011.
\newblock \href {http://dx.doi.org/10.1134/S0032946011010030}
  {\path{doi:10.1134/S0032946011010030}}.

\bibitem{DBLP:conf/stoc/BackursI15}
Arturs Backurs and Piotr Indyk.
\newblock Edit distance cannot be computed in strongly subquadratic time
  (unless {SETH} is false).
\newblock {\em SIAM Journal on Computing}, 47(3):1087--1097, 2018.
\newblock \href {http://dx.doi.org/10.1137/15M1053128}
  {\path{doi:10.1137/15M1053128}}.

\bibitem{DBLP:conf/focs/ChakrabortyDGKS18}
Diptarka Chakraborty, Debarati Das, Elazar Goldenberg, Michal Kouck{\'{y}}, and
  Michael~E. Saks.
\newblock Approximating edit distance within constant factor in truly
  sub-quadratic time.
\newblock In {\em FOCS'18}, pages 979--990, 2018.
\newblock \href {http://dx.doi.org/10.1109/FOCS.2018.00096}
  {\path{doi:10.1109/FOCS.2018.00096}}.

\bibitem{reservoir}
Min-Te Chao.
\newblock {A general purpose unequal probability sampling plan}.
\newblock {\em Biometrika}, 69(3):653--656, 12 1982.
\newblock \href {http://dx.doi.org/10.2307/2336002}
  {\path{doi:10.2307/2336002}}.

\bibitem{DBLP:conf/cpm/Charalampopoulos18}
Panagiotis Charalampopoulos, Maxime Crochemore, Costas~S. Iliopoulos, Tomasz
  Kociumaka, Solon~P. Pissis, Jakub Radoszewski, Wojciech Rytter, and Tomasz
  Waleń.
\newblock Linear-time algorithm for long {LCF} with $k$ mismatches.
\newblock In {\em CPM'18}, pages 23:1--23:16, 2018.
\newblock \href {http://dx.doi.org/10.4230/LIPIcs.CPM.2018.23}
  {\path{doi:10.4230/LIPIcs.CPM.2018.23}}.

\bibitem{Dhagat:1992:PLQ:139404.139409}
Aditi Dhagat, Peter G\'{a}cs, and Peter Winkler.
\newblock On playing ``{T}wenty questions'' with a liar.
\newblock In {\em SODA'92}, pages 16--22, 1992.
\newblock URL: \url{http://dl.acm.org/citation.cfm?id=139404.139409}.

\bibitem{FischerPaterson}
Michael~J. Fischer and Michael~S. Paterson.
\newblock String matching and other products.
\newblock In {\em Complexity of Computation}, pages 113--125, 1974.

\bibitem{DBLP:journals/ipl/FlouriGKU15}
Tom{\'{a}}s Flouri, Emanuele Giaquinta, Kassian Kobert, and Esko Ukkonen.
\newblock Longest common substrings with $k$ mismatches.
\newblock {\em Information Processing Letters}, 115(6-8):643--647, 2015.
\newblock \href {http://dx.doi.org/10.1016/j.ipl.2015.03.006}
  {\path{doi:10.1016/j.ipl.2015.03.006}}.

\bibitem{10.1145/8307.8309}
Zvi Galil and Raffaele Giancarlo.
\newblock Improved string matching with $k$ mismatches.
\newblock {\em SIGACT News}, 17(4):52–54, March 1986.
\newblock URL: \url{https://doi.org/10.1145/8307.8309}, \href
  {http://dx.doi.org/10.1145/8307.8309} {\path{doi:10.1145/8307.8309}}.

\bibitem{LCSk}
Garance Gourdel, Tomasz Kociumaka, Jakub Radoszewski, and Tatiana
  Starikovskaya.
\newblock Approximating longest common substring with $k$ mismatches: Theory
  and practice.
\newblock In {\em CPM'2020}, pages 18:1--18:12, 2020.

\bibitem{DBLP:journals/ipl/Grabowski15}
Szymon Grabowski.
\newblock A note on the longest common substring with $k$-mismatches problem.
\newblock {\em Information Processing Letters}, 115(6-8):640--642, 2015.
\newblock \href {http://dx.doi.org/10.1016/j.ipl.2015.03.006}
  {\path{doi:10.1016/j.ipl.2015.03.006}}.

\bibitem{DBLP:journals/jcss/ImpagliazzoPZ01}
Russell Impagliazzo, Ramamohan Paturi, and Francis Zane.
\newblock Which problems have strongly exponential complexity?
\newblock {\em Journal of Computer and System Sciences}, 63(4):512--530, 2001.
\newblock \href {http://dx.doi.org/10.1006/jcss.2001.1774}
  {\path{doi:10.1006/jcss.2001.1774}}.

\bibitem{MR737400}
William~B. Johnson and Joram Lindenstrauss.
\newblock Extensions of {L}ipschitz mappings into a {H}ilbert space.
\newblock In {\em {Modern Analysis and Probability}}, volume~26 of {\em
  Contemporary Mathematics}, pages 189--206, 1984.

\bibitem{DBLP:journals/ibmrd/KarpR87}
Richard~M. Karp and Michael~O. Rabin.
\newblock Efficient randomized pattern-matching algorithms.
\newblock {\em IBM Journal of Research and Development}, 31(2):249--260, 1987.
\newblock \href {http://dx.doi.org/10.1147/rd.312.0249}
  {\path{doi:10.1147/rd.312.0249}}.

\bibitem{DBLP:journals/algorithmica/KociumakaRS19}
Tomasz Kociumaka, Jakub Radoszewski, and Tatiana Starikovskaya.
\newblock Longest common substring with approximately $k$ mismatches.
\newblock {\em Algorithmica}, 81(6):2633--2652, 2019.
\newblock \href {http://dx.doi.org/10.1007/s00453-019-00548-x}
  {\path{doi:10.1007/s00453-019-00548-x}}.

\bibitem{kmacs}
Chris{-}Andre Leimeister and Burkhard Morgenstern.
\newblock kmacs: the $k$-mismatch average common substring approach to
  alignment-free sequence comparison.
\newblock {\em Bioinformatics}, 30(14):2000--2008, 2014.
\newblock \href {http://dx.doi.org/10.1093/bioinformatics/btu331}
  {\path{doi:10.1093/bioinformatics/btu331}}.

\bibitem{Porat:09}
Benny Porat and Ely Porat.
\newblock Exact and approximate pattern matching in the streaming model.
\newblock In {\em FOCS'09}, pages 315--323, 2009.
\newblock \href {http://dx.doi.org/10.1109/FOCS.2009.11}
  {\path{doi:10.1109/FOCS.2009.11}}.

\bibitem{DBLP:journals/corr/abs-1803-00904}
Aviad Rubinstein.
\newblock Hardness of approximate nearest neighbor search.
\newblock In {\em STOC'18}, pages 1260--1268, 2018.
\newblock \href {http://dx.doi.org/10.1145/3188745.3188916}
  {\path{doi:10.1145/3188745.3188916}}.

\bibitem{DBLP:journals/moc/TaoCH12}
Terence Tao, Ernest Croot~III, and Harald Helfgott.
\newblock Deterministic methods to find primes.
\newblock {\em AMS Mathematics of Computation}, 81(278):1233--1246, 2012.
\newblock \href {http://dx.doi.org/10.1090/S0025-5718-2011-02542-1}
  {\path{doi:10.1090/S0025-5718-2011-02542-1}}.

\bibitem{DBLP:journals/jcb/ThankachanAA16}
Sharma~V. Thankachan, Alberto Apostolico, and Srinivas Aluru.
\newblock A provably efficient algorithm for the \emph{k}-mismatch average
  common substring problem.
\newblock {\em Journal of Computational Biology}, 23(6):472--482, 2016.
\newblock \href {http://dx.doi.org/10.1089/cmb.2015.0235}
  {\path{doi:10.1089/cmb.2015.0235}}.

\bibitem{DBLP:journals/jcb/ThankachanCLAA16}
Sharma~V. Thankachan, Sriram~P. Chockalingam, Yongchao Liu, Alberto Apostolico,
  and Srinivas Aluru.
\newblock {ALFRED:} {A} practical method for alignment-free distance
  computation.
\newblock {\em Journal of Computational Biology}, 23(6):452--460, 2016.
\newblock \href {http://dx.doi.org/10.1089/cmb.2015.0217}
  {\path{doi:10.1089/cmb.2015.0217}}.

\bibitem{Thankachan2017}
Sharma~V. Thankachan, Sriram~P. Chockalingam, Yongchao Liu, Ambujam Krishnan,
  and Srinivas Aluru.
\newblock A greedy alignment-free distance estimator for phylogenetic
  inference.
\newblock {\em BMC Bioinformatics}, 18(8):238, Jun 2017.
\newblock \href {http://dx.doi.org/10.1186/s12859-017-1658-0}
  {\path{doi:10.1186/s12859-017-1658-0}}.

\bibitem{SETH_survey}
Virginia~Vassilevska Williams.
\newblock On some fine-grained questions in algorithms and complexity.
\newblock In {\em ICM'18}, pages 3447--3487, 2018.
\newblock \href {http://dx.doi.org/10.1142/9789813272880_0188}
  {\path{doi:10.1142/9789813272880_0188}}.

\end{thebibliography}

\end{document}